\documentclass[11pt]{article}
\usepackage{amssymb,amsfonts, mathrsfs,amsmath,graphicx,enumerate}
\usepackage{setspace}
\usepackage{paralist}
\usepackage{amsthm,calc,graphicx,iwona,float,wrapfig,caption}
\usepackage[colorlinks=true,citecolor=black,linkcolor=black,urlcolor=blue]{hyperref}
\usepackage[noabbrev,capitalise]{cleveref}
\usepackage[lmargin=26mm,rmargin=26mm,bmargin=26mm,tmargin=26mm]{geometry}
\usepackage[numbers,sort&compress]{natbib}
\renewcommand{\baselinestretch}{1.25}
\usepackage[hang]{footmisc} 

\renewcommand{\thefootnote}{\fnsymbol{footnote}}	
\allowdisplaybreaks
\newcommand\DateFootnote{
\begingroup
\renewcommand\thefootnote{}
\footnote{
\today}
\setcounter{footnote}{0}
\vspace*{-5ex}
\endgroup}

\makeatletter
\renewcommand\section{\@startsection {section}{1}{\z@}%
                                   {-3ex \@plus -1ex \@minus -.2ex}%
                                   {2ex \@plus.2ex}%
                                   {\normalfont\large\bfseries}}
\renewcommand\subsection{\@startsection{subsection}{2}{\z@}%
                                     {-2.5ex\@plus -1ex \@minus -.2ex}%
                                     {1.5ex \@plus .2ex}%
                                     {\normalfont\normalsize\bfseries}}
\renewcommand\subsubsection{\@startsection{subsubsection}{3}{\z@}%
                                     {-2ex\@plus -1ex \@minus -.2ex}%
                                     {1ex \@plus .2ex}%
                                     {\normalfont\normalsize\bfseries}}
 \renewcommand\paragraph{\@startsection{paragraph}{4}{\z@}%
                                    {1.5ex \@plus.5ex \@minus.2ex}%
                                    {-1em}%
                                    {\normalfont\normalsize\bfseries}}
\renewcommand\subparagraph{\@startsection{subparagraph}{5}{\parindent}%
                                       {1.5ex \@plus.5ex \@minus .2ex}%
                                       {-1em}%
                                      {\normalfont\normalsize\bfseries}}
\makeatother

\newcommand{\msn}[1]{MR:\,\href{http://www.ams.org/mathscinet-getitem?mr=MR#1}{#1}}

\newcommand{\doi}[1]{doi:\,\href{http://dx.doi.org/#1}{#1}}

\theoremstyle{plain}
\newtheorem{thm}{Theorem}

\theoremstyle{definition}

\newcommand{\floor}[1]{\lfloor{#1}\rfloor}

\renewcommand{\leq}{\leqslant}

\hypersetup{
    bookmarks=true,         
    unicode=false,          
    pdftoolbar=true,        
    pdfmenubar=true,        
    pdffitwindow=true,      
    pdftitle={My title},    
    pdfauthor={Author},     
    pdfsubject={Subject},   
    pdfnewwindow=true,      
    pdfkeywords={keywords}, 
    colorlinks=true,       
    linkcolor=blue,          
    citecolor=blue,        
    filecolor=blue,      
    urlcolor=blue           
}

\begin{document}

{\Large\bfseries\boldmath Three dimensional graph drawing with fixed vertices and one bend per edge}

\DateFootnote

{\large 
David~R.~Wood\,\footnotemark[3]
}

\footnotetext[3]{School of Mathematical Sciences, Monash University, Melbourne, Australia (\texttt{david.wood@monash.edu}). \\
Research supported by the Australian Research Council.}

%
\bigskip

A \emph{three-dimensional grid-drawing} of a graph represents the vertices by distinct points in $\mathbb{Z}^3$ (called \emph{grid-points}), and represents each edge as a polyline between its endpoints with bends (if any) also at gridpoints, such that distinct edges only intersect at common endpoints, and each edge only intersects a vertex that is an endpoint of that edge. This topic has been previously studied in \citep{Dev06,BCMW-JGAA04,DujWoo-DMTCS05,Wismath-TR04,MW15}. We focus on the problem of producing such a drawing, where the vertices are fixed at given grid-points. This variant has been previously studied in \citep{MW15,MorinWood-JGAA04}. \citet{MW15} recently proved the following theorem:

\begin{thm}
\label{MW}
For every graph $G$ with $n$ vertices, given fixed locations for the vertices of $G$ in $\mathbb{Z}^3$, there is a 
three-dimensional  grid-drawing of $G$ with at most three bends per edge. 
\end{thm}

We prove the same result with one bend per edge. 

\begin{thm}
\label{main}
For every graph $G$ with $n$ vertices and $m$ edges, given fixed locations for the vertices of $G$ in $\mathbb{Z}^3$, there is a 
three-dimensional  grid-drawing of $G$ with one bend per edge. 
\end{thm}

\begin{proof}
Consider each edge $vw$ of $G$ in turn. Say $v=(a,b,c)$ and $w=(p,q,r)$ in $\mathbb{Z}^3$.  Choose $x \in \{a-1,a+1\}\setminus \{p\}$ and $y \in \{q-1,q+1\} \setminus \{b\}$. Let $L(v,w):= \{(x,y,z): z \in \mathbb{Z}\}$. Observe that $L(v,w)$ is contained in a vertical line, and every point in $L(v,w)$ is visible from both $v$ and $w$. That is, a segment from $v$ or $w$ to any point in $L(v,w)$ passes through no other point in $\mathbb{Z}^3$. Choose a point $(x,y,z)\in L(v,w)$ such that (1)  no vertex of $G$ is positioned at $(x,y,z)$, (2)  the segment between $v$ and $(x,y,z)$ does not intersect any already drawn edge segment, and (3) the segment between $w$ and $(x,y,z)$ does not intersect any already drawn edge segment. Rule (1) forbids less than $n$ points in $L(v,w)$. 
Note that no edge-segment is drawn as a vertical line by this algorithm. Thus each edge-segment that is already drawn intersects the vertical line containing $L(v,w)$ in at most one point. Hence rule (2) forbids at most one point in $L(v,w)$ for each edge-segment that is already drawn. In total, rule (2) forbids less than $2m$ points in $L(v,w)$. Similarly, rule (3) forbids less than $2m$ points in $L(v,w)$. Since $L(v,w)$ has infinitely many points, there is a point $(x,y,z)\in L(v,w)$ satisfying (1), (2) and (3). Draw $vw$ with one bend at $(x,y,z)$. Then $vw$ passes through no vertex and intersects no other edge (except of course at $v$ or $w$). 
\end{proof}

The \emph{volume} of a three-dimensional grid-drawing is the number of grid points in a minimum axis-aligned box that contains the drawing.  \citet{MW15} considered the volume of the drawing produced by \cref{MW} to be ``unconstrained'', although they did provide volume bounds for a different result with the vertices in the plane. \citet{MW15} state that ``the general 3D point-set embeddability problem in which the specified point-set is not constrained to a plane remains as an interesting open problem if the volume must be constrained.'' We now show that the drawings produced by \cref{main} have constrained volume. In fact, in a certain sense the  volume is optimal. 

Say the initial vertex set is contained in an $X \times Y \times Z$ bounding box, without loss of generality, $[1,X] \times [1,Y] \times [1,Z]$. Then for each edge, the algorithm may choose the bend point $(x,y,z)$ with $x\in[0,X+1]$ and $y\in[0,Y+1]$ and $z\in[1,\max\{Z,n+4m\}]$. Thus the drawing is contained in an $(X+2) \times (Y+2) \times \max\{ Z, n+ 4m\}$ bounding box. 

We now show that in a special case, this volume bound is best possible. Say $G=K_n$ with the vertices at $(1,0,0),\dots,(n,0,0)$. Using the above notation, $X=n$ and $Y=1$ and $Z=1$. The above volume upper bound is  $(X+2)(Y+2) \max\{ Z, n+ 4m\} \leq O(n^3)$. \citet{MorinWood-JGAA04} proved that every 1-bend drawing of an $n$-vertex graph $G$ with vertices fixed on a line has volume at least $kn/2$ where $k$ is the cutwidth of $G$. The cutwidth of $K_n$ equals $\floor{n^2/4}$. Thus the volume of any 1-bend drawing of $K_n$, with  these vertex locations, is at least $n^3/8$, which is within a constant factor of the above volume upper bound.

\def\soft#1{\leavevmode\setbox0=\hbox{h}\dimen7=\ht0\advance \dimen7
  by-1ex\relax\if t#1\relax\rlap{\raise.6\dimen7
  \hbox{\kern.3ex\char'47}}#1\relax\else\if T#1\relax
  \rlap{\raise.5\dimen7\hbox{\kern1.3ex\char'47}}#1\relax \else\if
  d#1\relax\rlap{\raise.5\dimen7\hbox{\kern.9ex \char'47}}#1\relax\else\if
  D#1\relax\rlap{\raise.5\dimen7 \hbox{\kern1.4ex\char'47}}#1\relax\else\if
  l#1\relax \rlap{\raise.5\dimen7\hbox{\kern.4ex\char'47}}#1\relax \else\if
  L#1\relax\rlap{\raise.5\dimen7\hbox{\kern.7ex
  \char'47}}#1\relax\else\message{accent \string\soft \space #1 not
  defined!}#1\relax\fi\fi\fi\fi\fi\fi}

\end{document}